\documentclass{article}
\usepackage{amsmath}
\usepackage{amsfonts}
\usepackage{amssymb}
\usepackage{graphicx}
\usepackage{fullpage}
\usepackage{amsthm}
\usepackage{mathtools}
\usepackage{color}
\usepackage{blkarray}
\usepackage{hyperref}
\usepackage[title]{appendix}
\theoremstyle{definition}
\newtheorem{theorem}{Theorem}

\newtheorem{note}[theorem]{Note}
\newtheorem{definition}{Definition}
\newtheorem{example}[definition]{Example}

\newtheorem{lemma}[theorem]{Lemma}

\renewcommand{\S}{S}
\newcommand{\SM}{\mathcal{S}}
\newcommand{\R}{R}

\renewcommand{\|}{\,||\,}

\makeatletter
\newcommand{\@chapapp}{\relax}%
\makeatother

        {\begin{list}
                {\noindent\makebox[0mm][r]{$\bullet$}}
                {\leftmargin=5.5ex \usecounter{enumi}
 		 \topsep=1.5mm \itemsep=-.75ex}
        }
        {\end{list}}

\title{A reaction network scheme which implements the EM algorithm}
\author{Muppirala Viswa Virinchi \and Abhishek Behera \and Manoj Gopalkrishnan\
\\India Institute of Technology Bombay, Mumbai, India\\{\{axlevisu, abhishek.enlightened, manoj.gopalkrishnan\}@gmail.com}}
\date{22 April 2018}

\begin{document}
\maketitle

\begin{abstract}
A detailed algorithmic explanation is required for how a network of chemical reactions can generate the sophisticated behavior displayed by living cells. Though several previous works have shown that reaction networks are computationally universal and can in principle implement any algorithm, there is scope for constructions that map well onto biological reality, make efficient use of the computational potential of the native dynamics of reaction networks, and make contact with statistical mechanics. We describe a new reaction network scheme for solving a large class of statistical problems including the problem of how a cell would infer its environment from receptor-ligand bindings. Specifically we show how reaction networks can implement information projection, and consequently a generalized Expectation-Maximization algorithm, to solve maximum likelihood estimation problems in partially-observed exponential families on categorical data. Our scheme can be thought of as an algorithmic interpretation of E. T. Jaynes's vision of statistical mechanics as statistical inference.
\end{abstract}

\section{Introduction}
Many statistical problems involve fitting an exponential family of probability distributions to some data~\cite{anderson1970sufficiency}. Fisher's method of Maximum Likelihood gives a prescription for the best fit: pick that parameter $\theta$ that maximizes the likelihood $\Pr[x\mid \theta]$ of generating the data $x$. In problems of practical interest, the data $x$ is rarely available in full. It is more common to want to maximize a likelihood $\Pr[s\mid \theta]$ where $s=\SM x$ is a low-dimensional linear projection of the data $x$. \footnote{This situation can arise because only a linear projection is observable. It can also happen because we require a rich family of probability distributions on the space of $s$ points, but don't want to give away the nice properties of exponential families. We can achieve both by imagining that our observation $s$ comes from projection from a data vector $x$ living in a higher-dimensional space, and then employ an exponential family of probability distributions on this higher-dimensional space.} The EM algorithm~\cite{dempster1977maximum} is one way to solve this class of problems. We describe a reaction network scheme that implements a geometric version of the EM algorithm~\cite{Amari2016} for exponential families, and linear projections. To fix ideas, consider this example.
\begin{example}\label{ex:running}
Consider a three-sided die with the three sides labeled $X_1,X_2,X_3$ respectively. Suppose the probabilities of the three outcomes depend on two hidden parameters $\theta_1,\theta_2$ according to
$\Pr[X_1\mid\theta_1,\theta_2] \propto\theta_1^2,\,\,\Pr[X_2\mid\theta_1,\theta_2]\propto\theta_1\theta_2,\,\,\Pr[X_3\mid\theta_1,\theta_2]\propto\theta_2^2$. Further suppose that the die is rolled many times by a referee who records the frequences $n_1,n_2,n_3$ of the three outcomes. The outcomes are not visible directly to us. The referee tells us some linear combinations of $n_1, n_2, n_3$, and this is the only information available to us. For example, suppose the referee tells us $s_1=n_1+n_2+n_3$, the total number of die rolls, and also $s_2=n_1+n_2$, the total number of times the die outcome was either $X_1$ or $X_2$. We may be interested in the probability $\Pr[x_1,x_2,x_3,\theta_1,\theta_2 \mid s_1, s_2]$ or the maximum likelihood estimator $\theta^*=\arg\sup_{\theta} \Pr[\theta_1,\theta_2\mid s_1,s_2]$. 

Let $y(\theta_1,\theta_2):=(\theta_1^2,\theta_1\theta_2,\theta_2^2)$. The EM algorithm finds a local minimum of $D((x_1,x_2,x_3) \| y(\theta_1,\theta_2))=$ $x_1\log (x_1/\theta_1^2) - x_1 + \theta_1^2 + x_2\log (x_2/\theta_1\theta_2) - x_2 + \theta_1\theta_2 + x_3\log (x_3/\theta_2^2) - x_3 + \theta_2^2$ where $x_1+x_2+x_3 = s_1$, $x_1+x_2 = s_2$, and $x_1,x_2,x_3,\theta_1,\theta_2 >0$. We minimize $D$ because its global minimum is related to the mod of the probability distribution of interest and to the maximum likelihood estimator. The algorithm proceeds by alternately minimizing $D$ over the space of points $x$ that are consistent with the observations while keeping $\theta$ fixed, then minimizing $D$ with respect to $\theta$ keeping $x$ fixed, and so on iteratively. It halts when we encounter a pair $(x^*,\theta^*)$ which is a fixed point of the iteration. 

Our main contribution is to describe and analyze a novel reaction network scheme that implements the EM algorithm. Below is a reaction network obtained by our scheme for the die example. The dynamics of this network implements a generalized EM algorithm and finds a local minimum of $D$. 
\begin{align*}
X_1&\to X_1 + 2\theta_1 &2\theta_1&\to 0 &X_2&\to X_2+\theta_1+\theta_2 &\theta_1+\theta_2&\to 0\
\\X_3&\to X_3+2\theta_2 & 2\theta_2&\to 0 &X_1 + \theta_2&\to X_2 + \theta_2 &X_2 + \theta_1&\to X_1+\theta_1
\end{align*}
If $(x(t),\theta(t))=(x_1(t),x_2(t),x_3(t),\theta_1(t),\theta_2(t))$ are solutions to the mass-action ODEs for this system then we show in Theorem~\ref{thm:em} that $\frac{dD(x(t) \| y(\theta(t)))}{dt}\leq 0$. 
\end{example}

Notice that the first six reactions change the numbers of the $\theta_1,\theta_2$ species while the $X_1,X_2,X_3$ species act purely catalytically. The last two reactions change the numbers of the $X_1,X_2,X_3$ species while the species $\theta_1,\theta_2$ act purely catalytically. This is a general feature of our reaction scheme. There are two subnetworks, one which changes only the $\theta$ species and is catalyzed by the $X$ species, and the other which changes only the $X$ species and is catalyzed by the $\theta$ species. The first subnetwork computes an \textbf{M-Projection}, and the second computes an \textbf{E-Projection}~(Definition~\ref{def:infoproj}). 

The last two reactions in our example compute an E-Projection, i.e., if $x(t)$ is a solution trajectory to the last two reactions when $\theta_1,\theta_2$ are held fixed, then $\frac{dD(x(t)\|y(\theta))}{dt}\leq 0$. We have described this scheme to compute the E-Projection previously in \cite{virinchi2017stochastic}. Subsection~\ref{subsec:eproj} summarizes this previous work, showing that the subreaction network of our scheme that changes only the $X$ species always has this property, and will find a global minimum over all $x$ compatible with the observations for the function $D(x\|y(\theta))$ when keeping $\theta$ fixed~(Theorem~\ref{thm:eproj}). 

Theorem~\ref{thm:eproj} can be thought of as exploiting a formal similarity between free energy in physics and relative entropy in information theory. We encode the dynamics of the system so that its free energy corresponds to the function that we want to minimize, while the system explores the same space as allowed by the optimization constraints. In this way, we design our chemical system to solve the desired optimization problem.

The first six reactions in our example compute the {M-Projection}, i.e., if $\theta(t)=(\theta_1(t),\theta_2(t))$ is a solution trajectory to the first six reactions when $x_1,x_2,x_3$ are held fixed, then $\frac{dD(x\|y(\theta(t))}{dt}\leq 0$. The first six reactions will find a global minimum for $D(x\|y(\theta))$ over all $\theta$ while keeping $x$ fixed. One of the authors has previously described a similar scheme in \cite{gopalkrishnan2016scheme}. The current scheme is subtly, but significantly, different. It employs fewer reactions, admits simpler proofs, and combines well with the E-Projection scheme to allow the two parts to compute an EM algorithm. This last crucial property was absent in the M-Projection scheme in \cite{gopalkrishnan2016scheme}.
We show all this in general in Subsection~\ref{subsec:mproj} and Theorem~\ref{thm:mproj} for the subreaction network of our scheme that changes only the $\theta$ species.

Functions of the form $D(x(t)\|x'(t))$ are known to be Lyapunov functions for Markov chains when $x(t)$ and $x'(t)$ are solutions to the Markov chain's Master equation~\cite{van1992stochastic}. For nonlinear reaction networks, in contrast, prior to this work, only functions of the form $D(x(t) \| q)$ have been known to be Lyapunov functions, where $q$ is a point of detailed balance for the reaction network. Our M-projection systems are the first class of examples of nonlinear reaction networks with Lyapunov functions of the form $D(x\|x'(t))$ with time dependence on the second argument. The discovery of such a class of reaction networks and Lyapunov functions is a key contribution in this paper.

When the E-Projection reaction network and the M-Projection reaction network evolve simultaneously, we get a continuous-time generalized EM algorithm, where both the $x$ coordinates and the $\theta$ coordinates are being updated continuously. We show in Subsection~\ref{subsec:em} that if $(x(t),\theta(t))$ is a solution trajectory to the reaction network then $\frac{dD(x(t)\|y(\theta(t))}{dt}\leq 0$, so that for a generic initial point the system eventually settles into a local minimum $(\hat x,\hat\theta)$ of $D(x\|y(\theta))$ with $x$ constrained to values consistent with the observations. 
\section{Preliminaries}
\textbf{Notation:} For $u=(u_1,u_2,\dots,u_n)\in\mathbb{R}^n$, define $e^u:=(e^{u_1},e^{u_2},\dots, e^{u_n})\in\mathbb{R}^n_{>0}$. For $x=(x_1,x_2,\dots,x_n)\in\mathbb{R}^n_{>0}$, define $\log x:=(\log x_1,\log x_2,\dots,\log x_n)$. Define $x^u = \prod_{i=1}^n x_i^{u_i}$.
For $S\subseteq\mathbb{R}^n$ and $\beta\in\mathbb{R}^n$, define $\beta+S:=\{\beta+x\mid x\in S\}$ and $e^S:=\{e^x\mid x\in S\}\subseteq\mathbb{R}^n_{>0}$. For a matrix $(a_{ij})_{m\times n}$, its $i$'th row will be denoted by $a_{i.}$ and its $j$'th column will be denoted by $a_{.j}$.
\subsection{Information Geometry}

Fix a countable set $I$. The \textbf{extended relative entropy} $D: \mathbb{R}_{\geq 0}^I \times \mathbb{R}_{\geq 0}^I \rightarrow [-\infty,\infty]$ is 
$
D(x\|y)\coloneqq \sum_{i\in I} x_i\log \left(\frac{x_i}{y_i}\right) -x_i + y_i
$ 
with the convention $0\log 0 = 0$  and $x\log 0 = -\infty$ when $x\neq 0$. 
\begin{note}\label{note:gibbs}
$D(x\|y) = \sum_{i\in I} y_i h(x_i/y_i)$ where $h(x) = x \log x - x + 1$. Since $h(x)$ is nonnegative for all $x\in\mathbb{R}_{\geq 0}$, it follows that $D(x\|y)\geq 0$ with equality iff $x = y$.
\end{note}
\begin{note}
If $\sum_{i\in I} x_i =\sum_{i\in I} y_{i}$, in particular if $x,y$ are probability distributions on $I$, then $D(x\|y)=\sum_{i\in I}x_i\log \left(\frac{x_i}{y_i}\right)$.
\end{note}
\begin{note}
If $x,y$ are Poisson distributions, i.e., $x_i = e^{-\lambda}\frac{\lambda^i}{i!}$ and $y_i = e^{-\mu}\frac{\mu^i}{i!}$ for $i\in\mathbb{Z}_{\geq 0}$ then $\sum_{i\in\mathbb{Z}_{\geq 0}} x_i\log\frac{x_i}{y_i} = D(\lambda\| \mu) = \lambda \log\frac{\lambda}{\mu} -\lambda + \mu$. More generally, the relative entropy between two distributions, each of which is a product of Poisson distributions, equals the extended relative entropy between their rate vectors.
\end{note}

We state the Pythagorean Theorem of Information Geometry~\cite[Theorem 1.2]{Amari2016} for our special case, and give the short proof for completeness.
\begin{theorem}[Pythagorean Theorem]\label{thm:pyth}
For all $P,Q,R\in\mathbb{R}^n_{>0}$, we have $(P-Q)\cdot(\log Q - \log R)=0$ iff $D(P\|Q) + D(Q\|R) = D(P\|R)$
\end{theorem}
\begin{proof}
$D(P\|Q) + D(Q\|R) - D(P\|R) = \sum_{i=1}^n P_i\log\frac{P_i}{Q_i} - P_i + Q_i + Q_i\log\frac{Q_i}{R_i} - Q_i + R_i  - P_i\log\frac{P_i}{R_i} + P_i - R_i=\sum_{i=1}^n (P_i-Q_i)(\log R_i - \log{Q_i})$
\end{proof}

\begin{definition}\label{def:infoproj}
An Exponential Projection or \textbf{E-Projection}~\cite{Amari2016} (also called Information Projection or I-Projection~\cite{csiszar2004information}) of a point $y\in\mathbb{R}^n_{\geq 0}$ to a set $X\subset\mathbb{R}^n_{\geq 0}$ is a point $x^*=\arg\min_{x\in X} D(x\|y)$. A Mixture Projection or \textbf{M-Projection} (or reverse I-projection) of a point $x\in\mathbb{R}^n_{\geq 0}$ to a set $Y\subseteq\mathbb{R}^n_{>0}$ is a point $y^*=\arg\min_{y\in Y} D(x\|y)$. 
\end{definition}
If $X$ is convex then the E-Projection $x^*$ is unique~\cite{csiszar2003information}. If $Y$ is log-convex (i.e., $\log Y$ is convex) then the M-Projection $y^*$ is unique~\cite{csiszar2003information}. We will be interested in E-Projections when $X$ is an affine subspace (and hence convex), and M-projections when $Y$ is an exponential family (and hence log-convex). Various problems in probability and statistics can be reduced to computing such projections~\cite{csiszar2003information}. Amari~\cite{Amari2016} has shown that an alternation of these two projections corresponds to the usual EM algorithm~\cite{dempster1977maximum}, and has further argued that various other algorithms in Machine Learning such as k-means clustering, belief propagation, boosting, etc. can be understood as EM.

Birch's theorem is a well-known theorem in the statistics and reaction networks communities~\cite[Theorem~1.10]{pachter2005algebraic},\cite{birch63}. Below we state an extension of Birch's theorem which applies to the extended KL-divergence function, and show the connection to Information Projection. Our contribution is to present the results in a form that brings out the geometry of the situation.

\begin{theorem}[Birch's theorem and Information Projection]\label{thm:birch}
Fix a positive integer $n$. Let $V\subseteq\mathbb{R}^n$ be an affine subspace and let $V^\perp = \{w \mid v\cdot w  = 0 \text{ for all } v\in V\}$ be the orthogonal complement of $V$ in $\mathbb{R}^n$. Then
\begin{enumerate}
\item\label{birch1} For all $\alpha\in\mathbb{R}^n_{>0}$, the intersection of the polytope  $(\alpha + V^\perp)\cap \mathbb{R}^n_{\geq 0}$ with the hypersurface $e^V$ consists of precisely one point $\alpha^*$ called the Birch point of $\alpha$ relative to $V$.
\item\label{eproj2} For every $\beta\in e^V$, the E-Projection of $\beta$ to the polytope $(\alpha + V^\perp)\cap\mathbb{R}^n_{\geq 0}$ is $\alpha^*$. In particular, this E-Projection is unique.
\item\label{mproj3} The M-projection of $\alpha$ to $e^V$ is $\alpha^*$. In particular, this M-projection is unique.
\end{enumerate}
\end{theorem}
\begin{proof}
${(1)}$ Fix $\alpha\in\mathbb{R}^n_{>0}$. We first prove uniqueness: suppose for contradiction that there are at least two points of intersection $\alpha^*_1,\alpha^*_2$ of the polytope  $(\alpha + V^\perp)\cap \mathbb{R}^n_{\geq 0}$ with the hypersurface $e^V$. Since $\alpha^*_1,\alpha^*_2\in e^V$, we have $\log\alpha_1^* - \log\alpha_2^*\in V$. Since $\alpha- \alpha_1^*\in V^\perp$, we have $(\alpha-\alpha_1^*)\cdot(\log\alpha_1^* - \log\alpha_2^*)=0$. Then by the Pythagorean theorem, $D(\alpha\|\alpha^*_1) = D(\alpha\|\alpha^*_2) + D(\alpha^*_2\| \alpha^*_1)$ which implies $D(\alpha\|\alpha^*_1)\geq D(\alpha\|\alpha^*_2)$. By a symmetric argument, $D(\alpha\|\alpha^*_2)\geq D(\alpha\|\alpha^*_1)$ and we conclude $D(\alpha\|\alpha^*_2)= D(\alpha\|\alpha^*_1)$. In particular, $D(\alpha^*_2\|\alpha^*_1)=0$ which implies $\alpha^*_1=\alpha^*_2$ by Note~\ref{note:gibbs}. 

To prove that there exists at least one point of intersection, and to show $(2)$, fix $\beta\in e^V$. We will show that the E-Projection $\alpha^*$ of $\beta$ to $(\alpha+V^\perp)\cap\mathbb{R}^n_{\geq 0}$ belongs to $e^V$. This point $\alpha^*$ exists since $D(x\|\beta)$ is continuous, and hence attains its minimum over the compact set $(\alpha+V^\perp)\cap\mathbb{R}^n_{\geq 0}$. Further, because $\alpha^*$ is an infimum, we need that $\lim_{\lambda\to 0} \frac{d f((1-\lambda)\alpha^* + \lambda \alpha)}{d\lambda}=0$. That is, $(\alpha - \alpha^*)\log\frac{\alpha^*}{\beta}=0$, which implies that $\alpha^*\in e^V$ since $\alpha$ could have been replaced by any other arbitrary point of $(\alpha+V^\perp)\cap\mathbb{R}^n_{>0}$. 

$(3)$ now follows because $\alpha^*\in e^V$ implies $D(\alpha\|\alpha^*) + D(\alpha^*\|\beta) = D(\alpha\|\beta)$ for all $\beta\in e^V$, hence $\alpha^*$ is the M-Projection of $\alpha$ to $e^V$.
\end{proof}

\subsection{Reaction Network Theory}
We recall some concepts from reaction network 
theory \cite{feinberg72chemical,horn72necessary,Fein79,Manoj_2011Catalysis,anderson2010product,virinchi2017stochastic}. 

Fix a finite set $S$ of species. An $S$-reaction, or simply a \textbf{reaction} when $S$ is understood from context, is a formal chemical equation
\[
	\sum_{X\in S} y_X X \rightarrow \sum_{X\in S}y'_XX
\]
where the numbers $y_X,y'_X\in\mathbb{Z}_{\geq 0}$ are the \textbf{stoichiometric coefficients} of species $X$ on the \textbf{reactant} side and \textbf{product} side respectively. We write this reaction more pithily as $y\to y'$ where $y,y'\in\mathbb{Z}^{\S}_{\geq 0}$. A \textbf{reaction network} is a pair $(S,R)$ where $R$ is a finite set of $S$-reactions. It is \textbf{reversible} iff $y\to y'\in\R$ implies $y'\to y\in\R$. 

Fix $n,n'\in\mathbb{Z}^S_{\geq 0}$. We say that $n\mapsto_\R n'$ 
iff there exists a reaction $y\to y' \in\R$ with $n-y\in\mathbb{Z}^S_{\geq 0}$ and $n' = n + y'-y$. The \textbf{reachability} relation $n\Rightarrow_\R n'$ is the transitive closure of $\mapsto_\R$.
The \textbf{forward reachability class} of $n_0\in\mathbb{Z}^\S_{\geq 0}$ is the set $\Gamma(n_0)=\{ n \mid n_0\Rightarrow_\R n\}$. The \textbf{stoichiometric subspace} $H_\R$ is the real span of the vectors $\{y'-y \mid y\to y'\in \R\}$. The \textbf{conservation class} containing $x_0\in\mathbb{R}^\S_{\geq 0}$ is the set $C(x_0)=(x_0 + H_\R)\cap\mathbb{R}^S_{\geq 0}$. A reaction network $(\S,\R)$ is \textbf{weakly reversible} iff for every reaction $y\to y'\in\R$, we have $y'\Rightarrow_R y$.

Fix a weakly reversible reaction network $(\S,\R)$. The \textbf{associated ideal} $I_{(\S,\R)}\subseteq \mathbb{C}[\{x_i\mid i\in S\}]$  is the ideal generated by the binomials $\{ x^y - x^{y'}\mid y\to y'\in\R\}$. A reaction network is \textbf{prime} iff its associated ideal is a prime ideal, i.e., for all $f, g\in \mathbb{C}[x]$, if $fg\in I$ then either $f\in I$ or $g\in I$. 

\begin{example}
The reaction network given by the reactions $2X\rightleftharpoons 2Y$ is not prime, as the associated ideal $I_1$ is generated by the binomial $x^2-y^2$ and $(x+y)(x-y)\in I_1$ but $(x+y)\notin I_1$ and $(x-y)\notin I_1$. The reaction network given by the reactions $2X\rightleftharpoons Y$ is prime since the associated ideal $I_2$ is generated by the irreducible binomial $x^2-y$, and if $fg \in I_2$ then either $x^2-y$ divides $f$ or $x^2-y$ divides $g$, that is either $f\in I_2$ or $g\in I_2$.
\end{example}

A \textbf{reaction system} is a triple $(\S,\R,k)$ where $(\S,\R)$ is a reaction network and $k:\R\to\mathbb{R}_{> 0}$ is called the \textbf{rate function}. It is \textbf{detailed balanced} iff it is reversible and there exists a point $q\in\mathbb{R}^\S_{>0}$ such that $k_{y\to y'}\,q^y    = k_{y'\to y}\, q^{y'}$ for every reaction $y\to y'\in\R$. A point $q\in\mathbb{R}^\S_{>0}$ that satisfies the above condition is called a \textbf{point of detailed balance}.
\begin{note}\label{note:dbisaffine}
The set $\{ \log q \mid q\text{ is a point of detailed balance for }(\S,\R,k)\}$ is the simultaneous solution set to the affine system of equations $(y-y')\cdot\log q=\log\frac{k_{y'\to y}}{k_{y\to y'}}$ for all $y\to y'\in R$, and hence constitutes an affine space.
\end{note}

\section{Main}
\begin{definition}\label{def:Bnetwork}
  Let $\S$ be a finite set, and let $\mathcal{B}=\{b_1,b_2,...,b_r\}\subseteq\mathbb{Z}^\S$ be a finite set of integer vectors . For $l=1$ to $r$, let $b_l^+,b_l^-\in\mathbb{Z}^\S_{\geq 0}$ be the positive part and negative part of $b_l$, i.e., 
	\[
	b_{lj}^+ = 
	\begin{cases}
	b_{lj}\text{ if }b_{lj}>0\
	\\0\text{ otherwise}
	\end{cases} \text{ and }b_{lj}^- = \begin{cases}
	-b_{lj}\text{ if }b_{lj}<0\
	\\0\text{ otherwise}
	\end{cases}\text{ for all }j\in S.
	\]
	Then the reaction network $(\S,\mathcal{R_B})$ \textbf{generated by} $\mathcal{B}$ is given by the reactions: $b_l^+\rightleftharpoons b_l^-$ for  $l\in\{1,\dots,r\}$. 
\end{definition}

\begin{example}
If $\S=\{X_1,X_2,X_3\}$ and $\mathcal{B}=\{(1,0,1),(-2,1,1),(1,1,-3)\}$ then $(S,R_\mathcal{B})$ is given by:
$X_1+X_3\rightleftharpoons 0,\, X_2+X_3 \rightleftharpoons 2X_1, \, X_1+X_2\rightleftharpoons 3X_3$\end{example}

\subsection{Reaction networks compute E-Projections}\label{subsec:eproj}
The following theorem shows that points of detailed balance correspond to E-Projections. Further, if a detailed balanced reaction network has no critical siphons then solutions to mass-action kinetics converge to the E-Projections. 

We recall the notion of critical siphon~\cite{Angeli2007598,Manoj_2011Catalysis}. A \textbf{siphon} in a reaction network $(S,R)$ is a set $T\subseteq S$ of species such that for every reaction $y\to y'\in R$, if there exists $i\in T$ such that $y'_i>0$ then there exists $j\in T$ such that $y_j>0$. In particular, if all siphon species are absent, then they remain absent in future. A siphon $T$ is \textbf{critical} iff there exist $x\in\mathbb{R}^S_{\geq 0}$ and $y\in\mathbb{R}^S_{>0}$ such that $x-y\in H_R$ and $\{i\mid x_i=0\}= T$. 

The significance of critical siphons is that their absence allows easy demonstration of a detailed balanced reaction network version of the Markov Chain Ergodic Theorem, which is known as the Global Attractor Conjecture~\cite{horn74dynamics,GeoGac}. We take care to construct reaction network schemes that avoid critical siphons, thus ensuring that our reaction network dynamics provably converges to the right answer. It appears that avoiding critical siphons also confers advantages in terms of rate of convergence. We discuss this further in Subsection~\ref{subsec:rate}.

\begin{theorem}\label{thm:db}
Fix a detailed balanced reaction system $(S,R,k)$ with point of detailed balance $y\in\mathbb{R}^S_{>0}$. Let $x(t)$ be a solution to the mass-action equations for $(S,R,k)$ with $x(0)\in\R^S_{\geq 0}$. Then 
\begin{enumerate}
\item There exists a unique point of detailed balance $x^*\in(x(0)+H_R)\cap\R^S_{> 0}$.  
\item $\frac{d D(x(t) \| y)}{dt}\leq 0$ with equality iff $x(t)$ is a point of detailed balance.
\item\label{thm:db.3} If $(S,R)$ has no critical siphons then $\lim_{t\to\infty} x(t)=x^*$.
\item\label{thm:db.4} The point $x^*$ is the E-Projection of $y$ to the polytope $(x(0)+H_R)\cap\mathbb{R}^S_{\geq 0}$.
\end{enumerate}
\end{theorem}
Parts (1), (2), (3) are well-known in the theory of chemical reaction networks. We include the proofs of (1) and (2) for completeness.
\begin{proof}
(1) follows from Note~\ref{note:dbisaffine} and Theorem~\ref{thm:birch}.\ref{birch1}. For (2) by explicit calculation note that 
\[
\frac{dD(x(t)\|y)}{dt} =\sum_{r\rightleftharpoons r' \in \R} (k_{r\to r'} x(t)^r - k_{r'\to r} x(t)^{r'})\log\frac{k_{r'\to r} x(t)^{r'}}{k_{r\to r'} x(t)^r}
\] where each summand is $\leq 0$, hence $dD/dt\leq 0$ with equality iff $x(t)$ is a point of detailed balance. (3) follows from \cite[Theorem~2]{angeli2007petri}. (4) follows from Theorem~\ref{thm:birch}.\ref{eproj2}
\end{proof}

We are now going to describe a reaction network scheme to compute the E-Projection of an arbitrary point in $\mathbb{R}^n_{>0}$  to an arbitrary polytope in $\mathbb{R}^n_{\geq 0}$. Significantly our scheme will only create detailed balanced reaction networks without critical siphons, allowing the use of Theorem~\ref{thm:db}. To show that the reation networks described by this scheme have no critical siphons, we will need a definition and two lemmas which employ concepts from the theory of binomial ideals, and will not be used elsewhere in the paper. A reader who is not particularly concerned about critical siphons can omit these lemmas and jump to Theorem~\ref{thm:eproj}.

Fix a positive integer $n$. A \textbf{sublattice} of $\mathbb{Z}^n$ is a subgroup of the additive group $\mathbb{Z}^n$. It is necessarily a free and finitely generated abelian group, and hence isomorphic to an integer lattice. A sublattice $L\subseteq\mathbb{Z}^n$ is \textbf{saturated} iff for all $k\in\mathbb{Z}\setminus\{0\}$ and $v\in\mathbb{Z}^n$, if $k v\in L$ then $v\in L$. 

\begin{example}
The sublattice $L_1 = \{(x_1,3x_2)\mid x_1,x_2\in\mathbb{Z}\}$ of $\mathbb{Z}^2$ is unsaturated since $(0,3)\in L_1$ but $(0,1)\notin L_1$, whereas $L_2 = \{(x,3x)\mid x \in \mathbb{Z}\}$ is saturated.
\end{example}

\begin{lemma}\label{lem:saturated}
Let $S$ be a finite set, and let $\mathcal{B}\subseteq\mathbb{Z}^S$ be a finite set of integer vectors. Then the reaction network $(S,\mathcal{R_B})$ generated by $\mathcal{B}$ is prime iff the sublattice $L_\mathcal{B}\subseteq\mathbb{Z}^S$ is saturated.
\end{lemma}
\begin{proof}
This follows from~\cite[Corollary~2.15]{miller2011theory}, taking for $A$ the matrix whose rows form a basis for the sublattice perpendicular to $L_\mathcal{B}$, so that $I_A$ becomes the associated ideal $I_{(S,R)}$. The assumption of saturation is used in identifying the perpendicular to the perpendicular with the original lattice.
\end{proof}

\begin{lemma}\label{lem:prime}
A prime weakly-reversible reaction network has no critical siphons
\end{lemma}
\begin{proof}
Follows from \cite[Theorems~4.1,5.2]{Manoj_2011Catalysis}
\end{proof}\
\textbf{E-Projection reaction network scheme:} Fix a positive integer $n\in\mathbb{Z}_{>0}$. Consider $x_0,y\in \mathbb{R}^n_{>0}$ and an $n$-column \textbf{sensitivity matrix} $\SM$ of integers. Let $H_{x_0}=\{x\in\mathbb{R}^n_{\geq 0} \mid \SM x=\SM x_0\}$. To compute the E-Projection $\hat{x}$ of $y$ to $H_{x_0}$, we first compute a basis $\mathcal{B}=\{b_1,b_2,\dots,b_r\}$ to the sublattice $(\ker \SM)\cap\mathbb{Z}^n$. Using this, we describe a reaction system as follows:
\begin{enumerate}
\item The set of species is $\mathfrak{X}=\{X_1,X_2,\dots,X_n\}$,
\item The set of reactions is $R_\mathcal{B}$,
\item The reaction rates are chosen so that $y$ is a point of detailed balance, i.e.,
$\displaystyle
\frac{k_{b_l^-\to b_l^+}}{k_{b_l^+\to b_l^-}}=y^{b_l}
$ for $l=1$ to $r$, where $b_l^-,b_l^+$ are as in Definition~\ref{def:Bnetwork}
\end{enumerate}
We obtain the following theorem.

\begin{theorem}\label{thm:eproj} 
Let $x(t)$ be a solution to the mass-action equations for the reaction system $(\mathfrak{X},R_\mathcal{B},k)$ described above with $x(0)=x_0$. Then $\hat{x}=\lim_{t\to\infty} x(t)$ exists and equals the E-Projection of $y$ to $H_{x_0}$. 
\end{theorem}
\begin{proof}
From Lemmas~\ref{lem:saturated} and \ref{lem:prime}, the reaction network $(S,R_\mathcal{B})$ has no critical siphons. From Theorem~\ref{thm:db}.\ref{thm:db.3} and \ref{thm:db}.\ref{thm:db.4}, the result follows.
\end{proof}

\begin{example}[contd. from Example~\ref{ex:running}]\label{ex:running2}
For the three sided die, $\SM=\left(\begin{array}{rrr}1&1&1\\1&1&0\end{array}\right)$ where the first row represents $s_1$, the total number of times the die is rolled by the referee, and the second row represents $s_2$, the total number of times the die comes up either $X_1$ or $X_2$. The vector $\left(\begin{array}{r}1\\-1\\0\end{array}\right)$ is a basis for $\ker\SM$. The corresponding E-Projection reaction network is $X_1  \rightleftharpoons  X_2$. If $y=(1/3,1/3,1/3)$ represents our prior belief about the die, i.e., that it is a fair die and all three outcomes are equally likely, then we can set all reaction rates to $1$, and concentrations evolve according to the differential equations $\dot{x}_1 = -\dot{x}_2 = x_2 -  x_1,\,\dot{x}_3=0$. The derivative $d D(x(t) \| y)/dt=(x_1 - x_2)\log(x_2 /x_1) \leq 0$, showing that the dynamics is moving the system towards the E-Projection. If $x(0) = (2,20,27)$ then the system reaches equilibrium at  $x_1 = 11$, $x_2 = 11$, and $x_3 = 27$ which is the E-Projection of $(1/3,1/3,1/3)$ to $H_{x_0}=\{x\mid \SM x = \SM x_0\}$. This is also the most likely outcome corresponding to the observations $s_1=49,s_2=22$.
\end{example}

\subsection{Reaction networks compute M-Projections}\label{subsec:mproj}
\textbf{M-Projection Reaction Network Scheme:} Fix positive integers $m,n\in\mathbb{Z}_{>0}$. Consider $x\in \mathbb{R}^n_{>0}$, and a matrix $A=(a_{ij})_{m\times n}$ of {\em nonnegative} integers. Let $\operatorname{Col}(A)=\{a_{.1},a_{.2},\dots,a_{.n}\}$ denote the columns of A. Fix a map $y_A:\mathbb{R}^m\to\mathbb{R}^n_{>0}$ sending $\theta\longmapsto (c_1\theta^{a_{.1}},c_2\theta^{a_{.2}},\dots,c_n\theta^{a_{.n}})$ where $c_1,c_2,\dots,c_n\in\mathbb{R}_{>0}$. To compute the M-Projection $\hat{y}$ of $x$ to $y_A(\mathbb{R}^m)$, we describe a reaction system as follows:
\begin{enumerate}
\item The set of species is $\Theta=\{\theta_1,\theta_2,\dots,\theta_m\}$,
\item The set of reactions is $R_{\operatorname{Col}(A)}=\{0\rightleftharpoons a_{.1},0\rightleftharpoons a_{.2},\dots,0\rightleftharpoons a_{.n}\}$,
\item The reaction rates are chosen so that $\frac{k_{0\to a_{.j}}}{k_{a_{.j}\to 0}} = \frac{x_j}{c_j}$ for $j=1$ to $n$.
\end{enumerate}
We obtain the following theorem.
\begin{theorem}\label{thm:mproj}
Let $\theta(t)$ be a solution to the mass-action equations for the reaction system $(\Theta,R_\mathcal{B},k)$ described above. Then
\begin{enumerate}
\item $\dot{\theta} = A(x - y_A\circ\theta(t))$.
\item $\frac{\dot{\theta_i}}{\theta_i(t)} = -\frac{\partial D(x \| y_A(\theta)}{\partial\theta_i}|_{\theta=\theta(t)}$.
\item $\frac{dD(x \| y_A\circ\theta(t))}{dt}\leq 0$ with equality iff $A(x - y_A\circ\theta(t))=0$.
\item The limit $\hat{\theta}=\lim_{t\to\infty} \theta(t)\in\mathbb{R}^m_{\geq 0}$ exists and $y_A(\hat{\theta})=\hat{y}$ is the M-projection of $x$ to $y_A(\mathbb{R}^m)$.
\end{enumerate}
\end{theorem}
\begin{proof}
(1) and (2) are easily verified by explicit calculation. (3) follows from (2) by the chain rule, since $\frac{dD(x\|y_A\circ\theta(t))}{dt}=\sum_{i=1}^m \frac{ \partial D(x\|y_A(\theta))}{\partial\theta_i}|_{\theta(t)}\cdot \dot{\theta}_i$\\ $=-\sum_{i=1}^m \frac{1}{\theta_i(t)}\left(\frac{ \partial D(x\|y_A(\theta))}{\partial\theta_i}|_{\theta(t)}\right)^2\leq 0$. Equality implies $\frac{\dot{\theta_i}}{\theta_i(t)}=0$ for all $i$, hence by (1) we have $\dot{\theta}=A(x-y_A\circ\theta(t))=0$. To see (4), note that the limit exists because $D$ is decreasing in time, and bounded from below, and $\dot{D}=0$ implies $\dot{\theta}=0$. The limit point $\hat{\theta}$ is the M-projection because $A(x-y_A(\hat{\theta}))=0$ implies $\hat{\theta}$ is the Birch point of $x$ relative to $\log (y_A(\mathbb{R}^m))$, from Theorem~\ref{thm:birch}.
\end{proof}

\begin{example}[contd. from Example~\ref{ex:running2}]\label{ex:running3}
For the three sided die, the design matrix is $A=\left(\begin{array}{ccc}2&1&0\\0&1&2\end{array}\right)$, and $y_A(\theta_1,\theta_2) = (\theta_1^2,\theta_1\theta_2,\theta_2^2)$. The corresponding  network is:
\begin{align*}
0&\xrightleftharpoons[1]{x_1}2\theta_1  & 0&\xrightleftharpoons[1]{x_2}\theta_1+\theta_2 &0&\xrightleftharpoons[1]{x_3} 2\theta_2  
\end{align*}
Suppose the die was rolled $49$ times and the outcomes were $x_1 =11$, $x_2 =11$ and $x_3 =27$ respectively. We get the differential equations
$
\dot{\theta_1} = 2(11 - \theta_1^2) + (11 - \theta_1\theta_2),\,\,\dot{\theta_2} = 2(27 - \theta_2^2) + (11 - \theta_1\theta_2)$
The derivative $d D(x(t) \| y)/dt=-\left(\dot{\theta_1}^2/\theta_1 + \dot{\theta_2}^2/\theta_2\right) \leq 0$  The system is stationary (but not detailed balanced) at $\hat{\theta}_1 = 3$ and $\hat{\theta}_2 = 5$. The M-Projection point is $(9,15,25)$, and $A y_A(\hat{\theta}) = Ax$.
\end{example}

\subsection{Reaction Networks implement a generalized EM algorithm}\label{subsec:em}
\textbf{EM Reaction Network Scheme:} Fix positive integers $m,n\in\mathbb{Z}_{>0}$. Consider $x_0\in \mathbb{R}^n_{>0}$ and an $n$-column matrix $\SM=(s_{ij})$ of integers. Let $H_{x_0}=\{x\in\mathbb{R}^n_{\geq 0} \mid \SM x=\SM x_0\}$.  

Fix a matrix $A=(a_{ij})_{m\times n}$ of {\em nonnegative} integers. Let $\operatorname{Col}(A)=\{a_{.1},a_{.2},\dots,a_{.n}\}$ denote the columns of A. Fix a map $y_A:\mathbb{R}^m\to\mathbb{R}^n_{>0}$ sending $\theta\longmapsto (c_1\theta^{a_{.1}},c_2\theta^{a_{.2}},\dots,c_n\theta^{a_{.n}})$ where $c_1,c_2,\dots,c_n\in\mathbb{R}_{>0}$.

To compute $(\hat x,\hat\theta)$ which is a local minimum of $D(x \| y_A(\theta))$ when $x\in H_{x_0}$, we first compute a basis $\mathcal{B}=\{b_1,b_2,\dots,b_r\}$ to the sublattice $(\ker \SM)\cap\mathbb{Z}^n$. Using this, we describe a reaction system $\operatorname{EM}(A,\mathcal{B})$:
\begin{enumerate}
\item The set of species is $S=\mathfrak{X}\cup\Theta$ where $\mathfrak{X}=\{X_1,X_2,\dots,X_n\}$ and $\Theta=\{\theta_1,\theta_2,\dots,\theta_m\}$,

\item The reactions with rates are 
\begin{align*}
&X_j \xrightarrow{k_{a_{.j}}^+} X_j + \sum_{i=1}^m a_{ij}\theta_i\text{ and }\sum_{i=1}^m a_{ij}\theta_i\xrightarrow{k_{a_{.j}}^-}0&\text{ for }j=1\text{ to }n\
\\&\begin{rcases}
\displaystyle \sum_{j: b_{lj} >0} b_{lj} X_j + \sum_{i=1}^m d_{il}\theta_i\xrightarrow{k_l^+} \sum_{j: b_{lj} <0} -b_{lj} X_j + \sum_{i=1}^m d_{il}\theta_i\
\\\displaystyle\sum_{j: b_{lj}<0} -b_{lj} X_j + \sum_{i=1}^m e_{il}\theta_i\xrightarrow{k_l^-} \sum_{j: b_{lj}>0} b_{lj} X_j + \sum_{i=1}^m e_{il}\theta_i\
\end{rcases}&\text{ for $l=1$ to $r$}
\end{align*}

\item The reaction rates $k_{a_{.j}}^+$ and $k_{a_{.j}}^-$ are chosen so that
$k_{a_{.j}}^- = c_jk_{a_{.j}}^+$ for $j=1$ to $n$. A special choice is $k_{a_{.j}}^+=1$ and $k_{a_{.j}}^-=c_j$. The reaction rates $k_l^+,k_l^-$ and the stoichiometric coefficients $d_{il},e_{il}$ for the  reactions are chosen so that 
$
\displaystyle\frac{k_l^-}{k_l^+}\prod_{i=1}^m \theta_i^{e_{il}-d_{il}}=y(\theta)^{b_l}
$ for $l=1$ to $r$. A special choice is $k_l^-=y(1)^{b_l},k_l^+=1,e_{il} = \begin{cases}
a_{i.}\cdot b_l\text{ if }a_{i.}\cdot b_l>0\
\\0\text{ otherwise,}
\end{cases}$ and $d_{il} = \begin{cases}
-a_{i.}\cdot b_l\text{ if }a_{i.}\cdot b_l<0\
\\0\text{ otherwise,}
\end{cases}$ for $i=1$ to $m$, where $a_{i.}=(a_{i1},a_{i2},\dots,a_{in})$ is the $i$'th row of $A$.
\end{enumerate}

We obtain the following theorem.
\begin{theorem}\label{thm:em}
Let $(x(t),\theta(t))$ be a solution to the mass-action equations for the reaction system $\operatorname{EM}(A,\mathcal{B})$ described above with initial condition $(x(0),\theta(0))\in\mathbb{R}^\mathfrak{X}_{>0}\times\mathbb{R}^\Theta_{> 0}$. Then
\begin{enumerate}
\item $\frac{dD(x(t) \| y_A\circ\theta(t))}{dt}\leq 0$ with equality iff both 
	$x(t)$ is the E-Projection of $y_A\circ\theta(t)$ to $H_{x_0}$ and 
	$y_A\circ\theta(t)$ is the M-Projection of $x(t)$ to $y_A(\mathbb{R}^m)$.
\item The limit $(\hat{x},\hat{\theta})=\lim_{t\to\infty}(x(t), \theta(t))$ exists.
\item $\nabla_\theta D(x\|y_A(\theta))|_{\hat x,\hat\theta} = 0$ if $\hat\theta\in\mathbb{R}^\Theta_{>0}$.
\item $\nabla_x D(x\|y_A(\theta))|_{\hat{x},\hat\theta}$ is perpendicular to the stoichiometric subspace $H_{R_\mathcal{B}}$
\end{enumerate}
\end{theorem}
\begin{proof}
(1) From the chain rule, $\dot{D}(x(t) \| y_A\circ\theta(t)) = ( \nabla_x D\cdot \dot{x} + \nabla_\theta D\cdot\dot{\theta})|_{(x(t),\theta(t))}$. From Theorem ~\ref{thm:eproj}, the first term is nonpositive with equality iff $x(t)$ is the E-Projection of $y(\theta(t))$ onto $H_{x_0}$. From Theorem ~\ref{thm:mproj}, the second term is nonpositive with equality iff $y(\theta (t))$ is the M-Projection of $x(t)$ onto $y_A(\mathbb{R}^m)$. Hence $dD(x(t) \| y_A\circ\theta(t))/dt \leq 0$ with equality iff both $x(t)$ is the E-Projection of $y_A\circ\theta(t)$ to $H_{x_0}$ and $y_A\circ\theta(t)$ is the M-Projection of $x(t)$ to $y_A(\mathbb{R}^m)$.\
\\(2) Since $D(x(t) \| y_A\circ\theta(t))$ has a lower bound and $dD(x(t) \| y_A\circ\theta(t))/dt \leq 0$, eventually $dD(x(t) \| y_A\circ\theta(t))/dt = 0$ at which point, by the above argument, both the E-Projection and M-Projection subnetworks are stationary, so that $\dot{x}=0$ and $\dot\theta=0$. Hence the limit $(\hat{x},\hat{\theta})=\lim_{t\to\infty}(x(t), \theta(t))$ exists.\ 
\\(3) follows since  $\nabla_\theta D(x\|y_A(\theta))|_{\hat x,\hat\theta} = \dot{\theta}(t)/\theta(t)|_{\hat x,\hat{\theta}}=0$ when $\hat\theta\in\mathbb{R}^\Theta_{>0}$.\ 
\\(4) $\nabla_x D(x\|y_A(\theta))|_{\hat{x},\hat\theta} = \log \left(\frac{\hat{x}}{y_A(\hat\theta)}\right)$. By (1), the point $\hat{x}$ is the E-Projection of $y_A(\hat\theta)$ to $H_{x_0}$. Hence by Theorem~\ref{thm:birch}, the point $\hat{x}$ is the Birch point of $x_0$ relative to the affine space $\log y_A(\mathbb{R}^m)$, so that $(x - \hat x)\log\left(\frac{\hat{x}}{y_A(\hat\theta)}\right)=0$ for all $x\in H_{x_0}$. Hence the gradient $\nabla_x D(x\|y_A(\theta))|_{\hat{x},\hat\theta}$ is perpendicular to $H_{R_\mathcal{B}}$.
\end{proof}

\begin{example}[contd. from Example~\ref{ex:running3}]
For the three sided die, the design matrix is $A=\left(\begin{array}{ccc}2&1&0\\0&1&2\end{array}\right)$, $y_A(\theta_1,\theta_2) = (\theta_1^2,\theta_1\theta_2,\theta_2^2)$, $\SM=\left(\begin{array}{rrr}1&1&1\\1&1&0\end{array}\right)$ with basis $\left(\begin{array}{r}1\\-1\\0\end{array}\right)$ for $\ker\SM$. The corresponding EM reaction network is:
\begin{align*}
X_1&\to X_1 + 2\theta_1 &2\theta_1&\to 0 &X_2&\to X_2+\theta_1 + \theta_2 &\theta_1 + \theta_2&\to 0\
\\X_3&\to X_3+2\theta_2 & 2\theta_2&\to 0 &X_1 + \theta_2&\to X_2 + \theta_2 &X_2 + \theta_1&\to X_1+\theta_1
\end{align*}
With all reaction rates set to $1$, the concentrations evolve according to:  
\begin{align*}
\dot{\theta_1} &= 2(x_1 - \theta_1^2) + (x_2 - \theta_1\theta_2) &\dot{\theta_2} &= 2(x_3 - \theta_2^2) + (x_2 - \theta_1\theta_2)\\
\dot{x}_1 &= -\dot{x}_2 = \theta_1 x_2 - \theta_2 x_1  &\dot{x_3} &= 0
\end{align*}
The derivative $d D(x(t) \| y_A\circ\theta(t))/dt= (x_1 - x_2)\log(x_2/x_1)  - \dot{\theta}_1^2/\theta_1 -\dot{\theta}_2^2/\theta_2 \leq 0$. 

If $x(0) = (1,23,25)$ then irrespective of $\theta(0)$, the system reaches equilibrium at $
\hat\theta_1 = 3$, $\hat\theta_2 =5$, $\hat x_1 = 9$, $\hat x_2 = 15,\hat x_3 = 25$. The MLE $\hat\theta=\arg\sup_{\theta} \Pr [s_1, s_2 \mid \theta ] = \arg\sup_{\theta} \sum_{i =0}^{24}{\binom{49}{i,24-i,25}}
\theta_1^{24+i}\theta_2^{74-i} $ when maximized analytically through gradient 
descent converges to 
$\theta_1 = 0.42007781$ and $\theta_2 = 0.70013016$, which is proportional to 
$(3,5)$ upto numerical error. Since the likelihood doesn't change when the parameters are 
multiplied by the same factor, in this example the EM CRN has indeed found the MLE. The following graph shows how concentrations change through time to approach the steady state.
\[
	\includegraphics[width=5in]{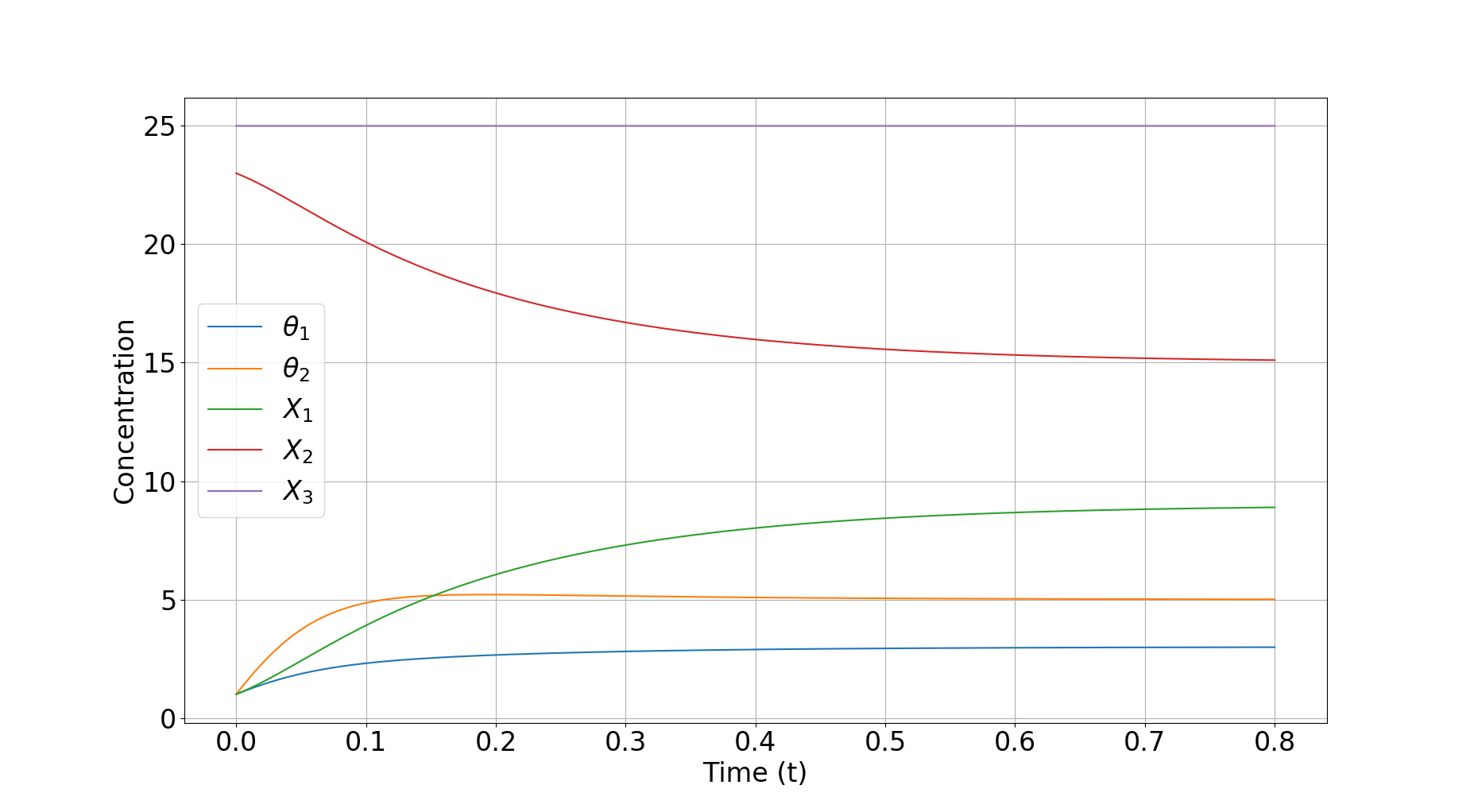}
\]
\end{example}

\begin{example}\label{ex:multipleequilibria}
Consider $A=\left(\begin{array}{ccc}2&1&0\\0&1&2\end{array}\right)$ and $\mathcal{S}=\left(\begin{array}{rrr}1&0&1\\1&1&1\end{array}\right)$.
The vector $\left(\begin{array}{r}1\\0\\-1\end{array}\right)$ spans $\ker \mathcal{S}$. The corresponding reaction network is 
\begin{align*}
X_1&\to X_1 + 2\theta_1 &2\theta_1&\to 0 &&X_2\to X_2+\theta_1+\theta_2 &\theta_1+\theta_2&\to 0\\
X_3&\to X_3+2\theta_2 & 2\theta_2&\to 0 &&X_1 + 2\theta_2\to X_3 + 2\theta_2 &X_3 + 2\theta_1&\to X_1 + 2\theta_1
\end{align*}
Here the concentration of $X_2$ remains invariant with time. Let $c$ be the initial concentration of $X_2$. If $c < 1/3$ then the system admits two stable equilibria and one unstable equilibrium. The points  $(y_1,c,y_2,\sqrt{y_1},\sqrt{y_2})$ and $(y_2,c,y_1,\sqrt{y_2},\sqrt{y_1})$ are the stable equilibria where $y_1=\frac{1-c}{2}+\frac{\sqrt{(1-3c)(1+c)}}{2}$ and $y_2=\frac{1-c}{2}-\frac{\sqrt{(1-3c)(1+c)}}{2}$, and $\left(\frac{1-c}{2},c,\frac{1-c}{2},\sqrt{\frac{1}{3}},\sqrt{\frac{1}{3}}\right)$ is the unstable equilibrium. On the other hand, if $c\geq 1/3$ then there is only one equilibrium point at $\left(\frac{1-c}{2},c,\frac{1-c}{2},\sqrt{\frac{1}{3}},\sqrt{\frac{1}{3}}\right)$, and this point is stable.
\end{example}

\begin{example}\label{ex:extinction}
Consider $A=\left(\begin{array}{ccc}2&1&0\\0&1&2\\\end{array}\right)$ and $\mathcal{S}=\left(\begin{array}{rrr}1&-1&0\\1&1&1\end{array}\right)$. The vector $\left(\begin{array}{r}1\\1\\-2\end{array}\right)$ spans $\ker \mathcal{S}$. The corresponding reaction network is 
\begin{align*}
&X_1\to X_1 + 2\theta_1, &&2\theta_1\to 0, &&X_2\to X_2+\theta_1+\theta_2, &&X_1 + X_2+3\theta_2\to 2X_3 +3\theta_2\\
&X_3\to X_3+2\theta_2, && 2\theta_2\to 0, &&\theta_1+\theta_2\to 0,&& 2X_3 + 3\theta_1\to X_1 + X_2+3\theta_1
\end{align*}
Here the set $\{X_1,X_2,\theta_1\}$ is a critical siphon. If we start at the initial concentrations $x_1=0.05,x_2=0.05,x_3=0.9,\theta_1=0.1,\theta_2=1.0$, then the system converges to $x_1=0,x_2=0,x_3=1,\theta_1=0,\theta_2=1$, hence this system is not persistent. This provides one explanation for this data: all the outcomes were of type $X_3$. If instead we start at $\theta_1=0.5,\theta_2=1.0$ and the same $x$ concentrations, then the system converges to $x_1=x_2=x_3=1/3,\theta_1=\theta_2=1/\sqrt{c}$. This provides a different explanation for the same data: all three outcomes have occurred equally frequently.
\end{example}

\begin{example}\label{ex:boltzmannmachine}
Boltzmann machines are a popular model in machine learning. Formally a Boltzmann machine is a graph $G=(V,E)$, each of whose nodes can be either $1$ or $0$. One associates to every configuration $s \in \{0,1\}^V$ of the Boltzmann machine an energy $E(s)=-\sum_i b_i s_i-\sum_{ij}w_{ij}s_is_j$. The probability of the Boltzmann machine being in configuration $s$ is given by the exponential family $P(s;b,w)\propto \exp(-E(s))$. Boltzmann machines can be used to do inference conditioned on partial observations, and learning of the maximum likelihood values of the parameters $b_i,w_{ij}$ can be done by a stochastic gradient descent.

Our EM scheme can be used to implement the learning rule of arbitrary Boltzmann machines in chemistry. We illustrate the construction on the 3-node Boltzmann machine with $V=\{x_1,x_2,x_3\}$:
\[
	\includegraphics[width=1.5in]{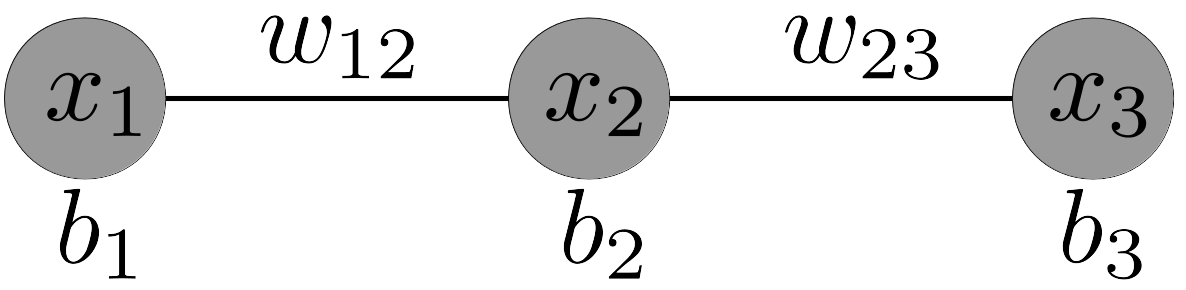}
\]
with biases $b_1,b_2,b_3$ and weights $w_{12}$ and $w_{23}$. We will work with parameters $\theta_i=\exp(b_i)$ and $\theta_{ij}=\exp(w_{ij})$. The design matrix $A=(a_{ij})_{5\times 8}$ is
\[A=\begin{blockarray}{cccccccccc}
    &&X_{000}&X_{001}&X_{010}&X_{011}&X_{100}&X_{101}&X_{110}&X_{111}\\
    \begin{block}{cc[cccccccc]}
      \theta_1&&0&0&0&0&1&1&1&1\\
      \theta_2&&0&0&1&1&0&0&1&1\\
      \theta_3&&0&1&0&1&0&1&0&1\\
      \theta_{12}&&0&0&0&0&0&0&1&1\\
      \theta_{23}&&0&0&0&1&0&0&0&1\\
    \end{block}
  \end{blockarray}\]
and the corresponding exponential model $y_A:\mathbb{R}^5\to\mathbb{R}^8_{>0}$ sends $\theta = (\theta_1,\theta_2,\theta_3,\theta_{
12},\theta_{23})$ to  $(\theta^{a_{.1}},\theta^{a_{.2}},\dots,\theta^{a_{.8}})$. If the node $x_2$ is hidden then the observation matrix $\mathcal S$ is
\[\mathcal{S}=\begin{blockarray}{cccccccc}
    X_{000}&X_{001}&X_{010}&X_{011}&X_{100}&X_{101}&X_{110}&X_{111}\\
    \begin{block}{[cccccccc]}
      1&0&1&0&0&0&0&0\\
      0&1&0&1&0&0&0&0\\
      0&0&0&0&1&0&1&0\\
      0&0&0&0&0&1&0&1\\
    \end{block}
  \end{blockarray}\]
Our EM scheme yields the reaction network
\begin{align*}
&\begin{rcases}
X_{ijk}\to X_{ijk} + i\theta_1+j\theta_2+k\theta_3+ij\theta_{12}+jk\theta_{23},\
\\i\theta_1+j\theta_2+k\theta_3+ij\theta_{12}+jk\theta_{23}\to 0
\end{rcases} \text{ for }i,j,k=0,1\\
&\begin{rcases}
X_{i1k}\to X_{i0k}\
\\ X_{i0k} +\theta_2+i\theta_{12}+k\theta_{23}\to X_{i1k} +\theta_2+i\theta_{12}+k\theta_{23}
\end{rcases}\text{ for }i,k=0,1
\end{align*}

Suppose we observe a marginal distribution $(0.24,0.04,0.17,0.55)$ on the visible nodes $x_1,x_3$. To solve for the maximum likelihood $\hat\theta$, we can initialize the system with $X_{000}=0.24,X_{001}=0.04,X_{010}=0,X_{011}=0,X_{100}=0.17,X_{101}=0.55,X_{110}=0,X_{111}=0$ and all $\theta$'s initialized to $1$, the system reaches steady state at $\hat\theta_1=0.5176,\hat\theta_2=0.0018,\hat\theta_3=0.3881,\hat\theta_{12}=0.8246,\hat\theta_{23}=0.7969, \hat X_{000}=0.2391,\hat X_{001}=0.0389, \hat X_{010}=0.0009, \hat X_{011}=0.011, \hat X_{100}=0.1695, \hat X_{101}=0.5487, \hat X_{110}=0.0005, \hat X_{111}=0.0013$.
\[
	\includegraphics[width=5in]{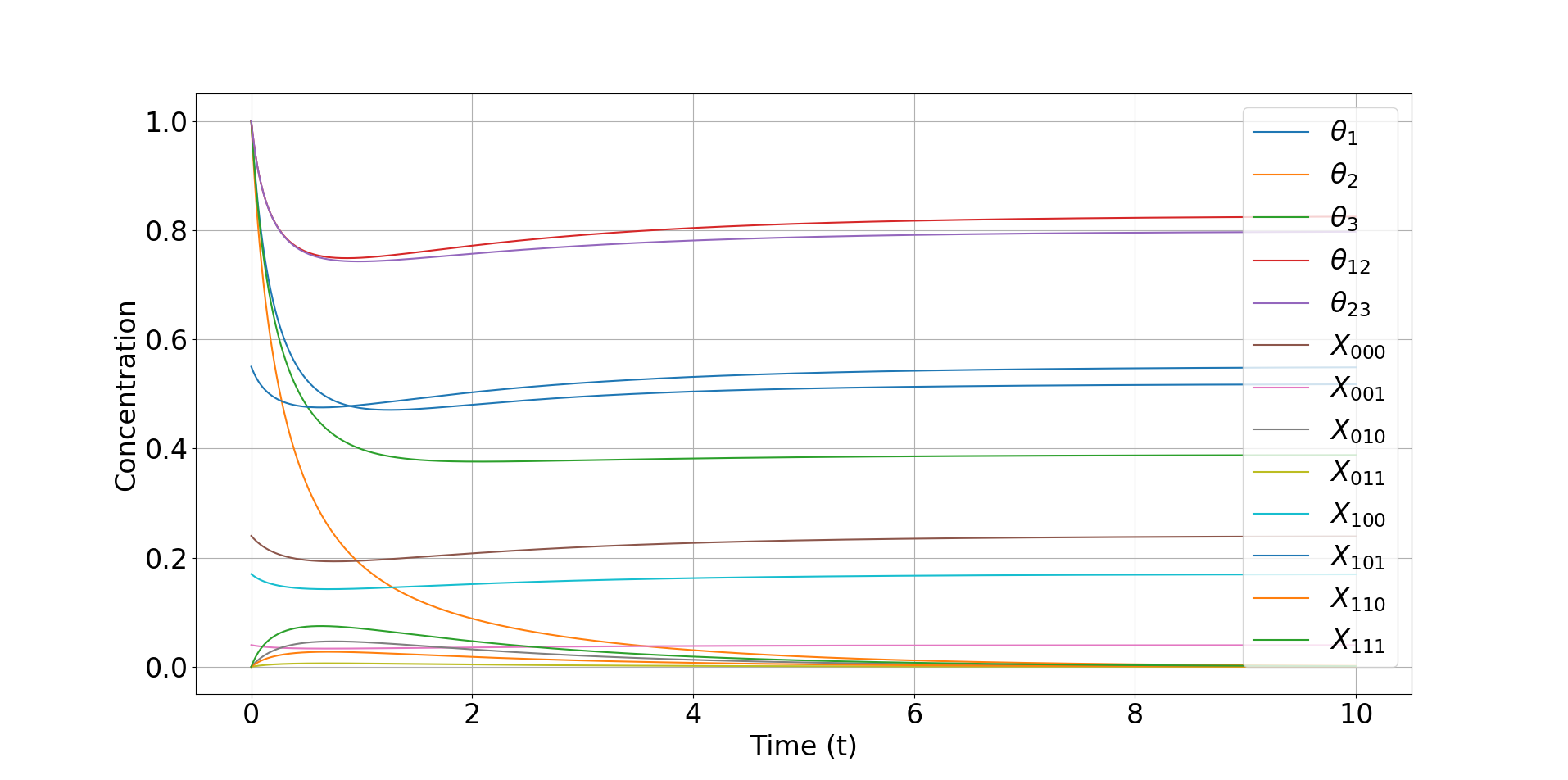}
\]

\end{example}

\section{Related work}
There is a rich history of theoretical and empirical results showing that reaction networks can perform computations~\cite{hjelmfelt1991chemical,Klavins_2011Biomolecular,Winfree_2011Neural,Winfree_2011SquareRoot,sarpeshkar2014analog,napp2013message,buisman2009computing,cardelli2018chemical,soloveichik2008computation,chen2014deterministic}. Typically these results take a known algorithm, and show how to implement it with chemical reaction networks. In contrast, we have obtained what appears to be a new algorithm within the broad class of generalized EM algorithms. Our algorithm is {\em natural} in the sense that it was suggested by the mathematical structure of reaction network dynamics itself, so that analysis of our scheme proceeds from insights about reaction network kinetics rather than from insights about the behavior of some existing classical algorithm. 

Similarities in the mathematical structure of statistics, statistical mechanics, and information theory have been noticed and remarked upon several times \cite{tribus1971energy,csiszar2004information,wainwright2008graphical,zellner1988optimal}, and have led to multiple contributions~\cite{wiener1961cybernetics,cencov2000statistical,Amari2016,mackay2003information,jaynes1957information,baez2012algorithmic} with the goal of presenting some or all of these topics from ``the point of view from which the subject appears in the greatest simplicity,'' to borrow a prescient quote from J. W. Gibbs. Our EM reaction network scheme employs a statistical mechanical system to minimize an information theoretic quantity in the service of solving a statistical problem. It is a concrete illustration of the connections between these three disciplines, and of the opportunities that lie at their intersection.

We now compare our scheme with three other schemes that show how to implement machine learning algorithms with reaction networks.

The belief propagation scheme of Napp and Adams~\cite{napp2013message} shows how reaction networks can implement sum-product algorithms from probabilistic graphical models to compute marginals of joint distributions. There is some formal similarity between the reaction networks of Napp and Adams, and our EM reaction network. In particular, their scheme also has two sets of species, ``sum'' and ``product'' species, and two subnetworks. In each subnetwork, one set of species changes in number, and the other set appears catalytically. We speculate that this may be because message passing algorithms are themselves a special case of the EM algorithm~\cite{ikeda2004stochastic,ikeda2004information}, in which case possibly the Napp-Adams scheme may be related in some as-yet-undiscovered way to our EM scheme.

In a brilliant paper, Zechner et al.~\cite{zechner2016molecular} show how to implement a Kalman filter with reaction networks, and implement this scheme {\em in vivo} in E. Coli. Their approach is to write down a dynamical system describing the filter, then change variables if required so the dynamical system falls within the class of systems that are implementable with CRNs, and finally obtain a DNA strand displacement implementation for the dynamical system. Their work shows that filters, being self-correcting, are robust, and can tolerate some amount of model mismatch. Further, such systems when implemented {\em in vitro} and {\em in vivo} work as advertised. This is very encouraging for the empirical prospects of our schemes.

There appears to be a certain degree of art involved in Zechner et al.'s setting to get the right change of variables which makes the dynamical system implementable by CRNs. In comparison, the information processing task directly informs our CRN architecture. Since Hidden Markov Models (HMMs) are special cases of exponential families as well, our EM reaction network can in principle be extended to implementations of the forward-backward algorithm for HMMs, which is closely related to the Kalman filter.

In~\cite{poole2017chemical}, Poole et al. have shown how to implement a Boltzmann machine with reaction networks. The reaction network is able to do inference, but the Boltzmann machine training to learn weights has to be done {\em in silico}. In Example~\ref{ex:boltzmannmachine}, we have shown how the EM algorithm can also be used to implement Boltzmann machines. There are pros and cons to our EM approach for this problem. The advantage is that Boltzmann machine training also happens {\em in vitro} and in an online manner. The disadvantage is that as described the EM implementation has not exploited the graphical structure of Boltzmann machines, and hence requires an exponentially large number of species for implementation. In contrast, the Poole et al. construction requires a linear number of species.

\section{Discussion}\label{sec:discuss}
\subsection{Rate of convergence}\label{subsec:rate}
Speed is a key aspect of the analysis of any algorithm. We would like to be able to say that every mass-action trajectory $x(t)$ to a reaction network described by our scheme converges exponentially fast to the stationary state $x^*$, i.e. there exists $T>0$ such that for all $\epsilon>0$, for all $x(0)\in\mathbb{R}^S_{\geq 0}$, for all $\tau\geq T \log \frac{d(x(0),x^*)}{\epsilon}$ the distance $d(x(\tau),x^*)\leq \epsilon$. We would like to say this for the E-projection system, the M-projection system, and the EM system. The E-projection system has the nicest structure, being detailed balanced, and hence is the first candidate for showing such a result. Even here, the best available result appears to be slightly weaker: Desvillettes et al.~\cite{desvillettes2017trend} have shown that if a detailed balanced system has no boundary equilibria then for all $x(0)\in\mathbb{R}^S_{\geq 0}$, there exists $T>0$ such that for all $\epsilon>0$, for all $\tau\geq T \log \frac{d(x(0),x^*)}{\epsilon}$ the distance $d(x(\tau),x^*)\leq \epsilon$. The gap is that here $T$ may depend on $x(0)$. We conjecture that exponential convergence should be true for the E-projection and M-projection systems. One could try to prove this via an entropy production inequality: there exists $\lambda>0$ such that $-\dot D(t) \geq\lambda D(t)- D^*$ where $D^*$ is the value of $D$ at the stationary point.

\subsection{A proposal for how a biological cell infers its environment}
Biological cells are capable of identifying and responding to the environment from the information provided to them by transmembrane receptors. Given partial observations, biochemical reaction networks have to identify the most likely environment that could have caused these observations. Could these networks be employing some variant of our EM scheme? We attempt to sketch how elements of our scheme could map onto the biological situation. 

A single receptor may be sensitive in different degrees to multiple ligands, so the information conveyed into the cell from the activation of a receptor may not be unlike a message sent by the referee in Example~\ref{ex:running}. Each type of transmembrane receptor corresponds to a row of the sensitivity matrix $\SM$. Species $X_1,X_2,\dots,X_n$ are ``shadow'' species inside the cell whose job is to track the concentration of the corresponding ligand species outside. Presumably this matrix $\SM$ is ``known'' to the cell, in the sense that it is hardcoded in some way into the structure of the biochemical reaction network.
 
The parameters $\theta$ to be estimated may correspond to certain underlying environmental variables that need to be monitored by the cell. For example, $\theta_1$ may encode for the threat level in the environment, and $\theta_2$ for the food level. These parameters may not be directly accessible, instead they have to be inferred indirectly from proxy measurements. The proxies are the ligand molecules whose concentrations are known to obey some law which depends on these parameters. The precise form of the law is described by the design matrix $A$ through the map $y_A$, and has presumably been learnt by the cell through evolution, so that it is hard-coded into the structure of the biochemical reaction network. 

In this setting, the EM reaction network behaves like an online algorithm. As new information streams into the cell, the reaction network dynamics tracks the current state of the environment, making the necessary modifications to the concentrations of the $\theta$ and $X$ species. 

Biochemical circuits in living cells, which have had the advantage of several billion years of evolution, are likely to be far more sophisticated than the scheme we have suggested. Nevertheless, we would like to think that our scheme can be a starting point for the study of the actual schemes that cells employ, to initiate a conversation between the analytical approach of systems biology and the synthetic biology approach.

\subsection{Stochastic Aspects of our Scheme}
For pedagogical reasons, we have focused in this paper on deterministic dynamics. We now indicate certain stochastic aspects of our scheme. We have shown in \cite{virinchi2017stochastic} that at steady state, the E-Projection reaction system samples from the Bayesian posterior distribution $\Pr[x \mid s]$ when the system is evolved according to stochastic mass-action kinetics. In this sense, the stochastic dynamics is the ``right'' dynamics for E-Projection. In contrast, we do not have a similar result for stochastic dynamics in the M-projection system. We appear to need the large-volume limit here, so that deterministic mass-action kinetics is a good description of the system's evolution.

This suggests that the best way to implement this EM scheme using molecules would be to employ a hybrid dynamics, whereby each unit volume in an infinite volume vessel implements stochastic E-Projection system, with the $X$ species localized in these unit volumes, whereas the $\theta$ species can move freely in the entire infinite volume, so that their dynamics is correctly described by deterministic mass action. This could be biologically meaningful if the localized molecules around different receptors can be related to the $X$ species, and the molecules that move freely in the cytosol and interact with different receptors can be related to the $\theta$ species. Note that the volume for such hybrid dynamics need not be truly infinite, since even copy numbers as low as $20$ typically lead to a dynamics which is well described by deterministic mass-action kinetics.

\bibliographystyle{plain}
\bibliography{../../eventsystems}
\end{document}